\documentclass[pra,twocolumn,superscriptaddress,nofootinbib]{revtex4}
\usepackage[letterpaper, left=0.8in, right=0.8in, top=0.9in, bottom=0.9in]{geometry}
\usepackage{titlesec}

\usepackage{cmap} % Load before fontenc
\usepackage[utf8]{inputenc}
\usepackage[english]{babel}
\usepackage[T1]{fontenc}
\usepackage{amsmath}
\usepackage{amsfonts}
\usepackage{bm}
\usepackage{xcolor}
\definecolor{blueviolet}{rgb}{0.2, 0.2, 0.6}
\definecolor{webgreen}{rgb}{0,.5,0}
\definecolor{webbrown}{rgb}{.6,0,0}
\usepackage[pdftex,
	bookmarks=false,
	colorlinks=true, %allcolors=blueviolet,
	urlcolor=webbrown,
	linkcolor=blueviolet,
	citecolor=webgreen,
	pdfstartpage=1,
	pdfstartview={FitH},  % FitBH
	bookmarksopen=false
	]{hyperref}
\usepackage{tikz}
\usepackage{natbib}
\usepackage{mathtools}
\usepackage{graphicx}
\usepackage{csquotes}
\usepackage{braket}
\usepackage{dsfont}
\usepackage{listings}
\allowdisplaybreaks
\DeclareMathOperator{\Expect}{\mathbb{E}}

\usepackage{multirow}
\usepackage{amssymb}

\usepackage{amsthm}

\usepackage{color}
\usepackage{nicefrac}

\DeclareFixedFont{\ttb}{T1}{txtt}{bx}{n}{9} % for bold
\DeclareFixedFont{\ttm}{T1}{txtt}{m}{n}{9}  % for normal

\usepackage{color}
\definecolor{deepblue}{rgb}{0,0,0.5}
\definecolor{deepred}{rgb}{0.6,0,0}
\definecolor{deepgreen}{rgb}{0,0.5,0}

\newcommand\pythonstyle{\lstset{
language=Python,
basicstyle=\ttm,
morekeywords={self},              % Add keywords here
keywordstyle=\ttb\color{deepblue},
emph={MyClass,__init__},          % Custom highlighting
emphstyle=\ttb\color{deepred},    % Custom highlighting style
stringstyle=\color{deepgreen},
frame=tb,                         % Any extra options here
showstringspaces=false
}}

\lstnewenvironment{python}[1][]
{
\pythonstyle
\lstset{#1}
}
{}

\newcommand\pythoninline[1]{{\pythonstyle\lstinline!#1!}}

\definecolor{orange}{RGB}{255,127,0}

\def\ket#1{\ensuremath{\mathinner{|{#1}\rangle}}}

\newcommand{\norm}[1]{\left\lVert#1\right\rVert}

\newcommand{\cA}{{\mathcal{A}}}

\newcommand{\cM}{{\mathcal{M}}}

\newtheorem{proposition}{Proposition}

\newtheorem{theorem}{Theorem}

\newtheorem{definition}{Definition}

\newtheorem{lemma}{Lemma}

\newtheorem*{criterion}{Criterion for classical access}

\newcommand{\Id}{I}

\usepackage[noend]{algpseudocode}
\usepackage{algorithm,algorithmicx}

\algrenewcommand\alglinenumber[1]{\sf\scriptsize\color{black}{#1}}
\algrenewcommand\algorithmicrequire{\textbf{Input:}}
\algrenewcommand\algorithmicensure{\textbf{Output:}}

\begin{document}

\title{Revisiting dequantization and quantum advantage in learning tasks}
\date{\today}
\author{Jordan Cotler}
\email{jcotler@fas.harvard.edu}
\affiliation{Harvard Society of Fellows, Cambridge, MA 02138 USA}
\affiliation{Black Hole Initiative, Cambridge, MA 02138 USA}
\author{Hsin-Yuan Huang}
\email{hsinyuan@caltech.edu}
\affiliation{Institute for Quantum Information and Matter, Caltech, Pasadena, CA, USA}
\affiliation{Department of Computing and Mathematical Sciences, Caltech, Pasadena, CA, USA}
\author{Jarrod R.~McClean}
\email{jmcclean@google.com}
\affiliation{Google Quantum AI, 340 Main Street, Venice, CA 90291, USA}

\begin{abstract}
It has been shown that the apparent advantage of some quantum machine learning algorithms may be efficiently replicated using classical algorithms with suitable data access -- a process known as dequantization.
Existing works on dequantization compare quantum algorithms which take copies of an $n$-qubit quantum state $|x\rangle = \sum_{i} x_i |i\rangle$ as input to classical algorithms which have sample and query (SQ) access to the vector $x$.
In this note, we prove that classical algorithms with SQ access can accomplish some learning tasks exponentially faster than quantum algorithms with quantum state inputs.
Because classical algorithms are a subset of quantum algorithms, this demonstrates that SQ access can sometimes be significantly more powerful than quantum state inputs.
Our findings suggest that the absence of exponential quantum advantage in some learning tasks may be due to SQ access being too powerful relative to quantum state inputs.
If we compare quantum algorithms with quantum state inputs to classical algorithms with access to measurement data on quantum states, the landscape of quantum advantage can be dramatically different.
We remark that when the quantum states are constructed from exponential-size classical data, comparing SQ access and quantum state inputs is appropriate since both require exponential time to prepare.
\end{abstract}

\maketitle

\section{Introduction}

Quantum machine learning (QML) \cite{biamonte2017quantum} studies how quantum computers can be used to solve learning tasks, and what the potential advantages of quantum machine learning (ML) algorithms over classical machine learning algorithms are.
Previously, several quantum ML algorithms were identified as having a plausible exponential advantage over classical ML algorithms.
Examples include quantum principal component analysis \cite{lloyd2014quantum}, quantum support vector machines \cite{rebentrost2014quantum}, quantum linear system solvers \cite{harrow2009quantum}, and others \cite{biamonte2017quantum}.
Despite the promise of these algorithms, their advantages over classical algorithms have not previously been established.  This raises concerns that the desired advantages may be illusory~\cite{aaronson2015read}.

Relatedly, a recent series of works \cite{tang2018quantum, tang2019quantum, gilyen2018quantum, chia2020sampling} show that under the assumption of a particular classical data access model, there exist quantum-inspired classical ML algorithms that only have polynomial slowdown compared to their quantum counterparts.
These classical algorithms are often called the \emph{dequantizations} of their corresponding quantum algorithms.
For context, many QML algorithms take copies of an $n$-qubit quantum state as input, i.e.~multiple copies of a state $\ket{x} = \sum_{i} x_i \ket{i}$ with $x \in \mathbb{C}^{2^n}$.
The literature on dequantized algorithms attempts to provide an analogous resource of classical data inputs for classical algorithms, with the goal of making the classical and quantum settings be on more equal footing.

In particular, a dequantized algorithm has a certain form of sample access to components of a vector $x \in \mathbb{C}^{2^n}$, and can read out any chosen component $x_i$ to arbitrarily high precision. This is known as \textit{sample and query} (SQ) \textit{access} to the vector $x$, which we will define more precisely later.

A key question is whether SQ access is the appropriate classical analog of quantum state inputs.  If it is, then the work of~\cite{tang2018quantum, tang2019quantum, gilyen2018quantum, chia2020sampling} means we should not view certain QML algorithms as having an significant advantage over their dequantized classical counterparts.  On the other hand, if SQ access is not the appropriate classical analog of quantum state inputs, then we should view QML algorithms as being incommensurate with their dequantized counterparts; accordingly, the computational powers of the former and the latter should not be compared.

\begin{figure*}[t]
    \centering
    \includegraphics[width=0.8\textwidth]{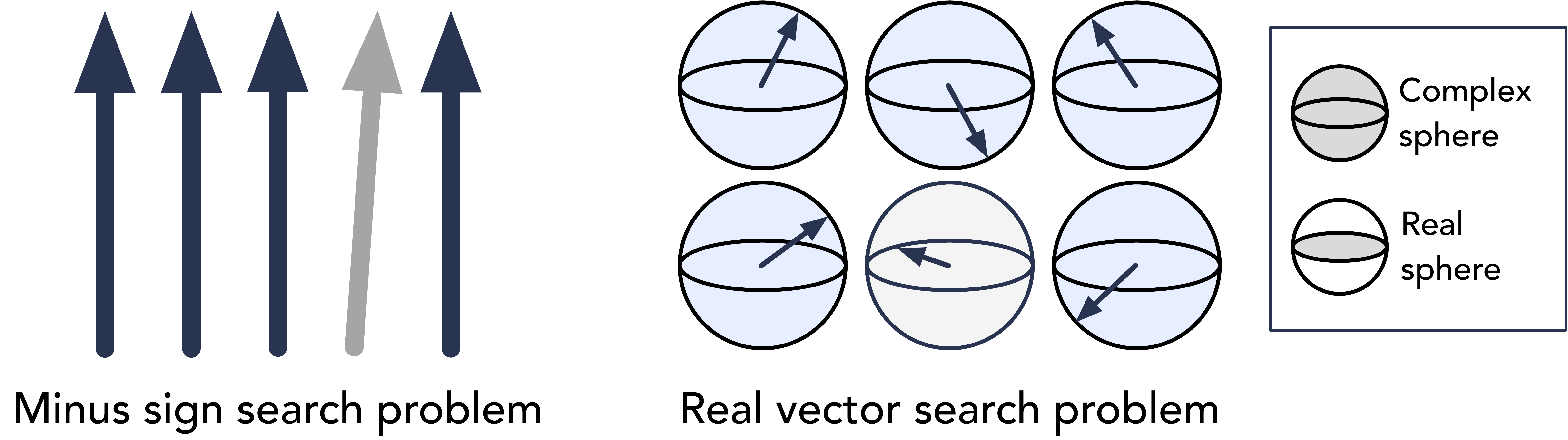}
    \caption{\emph{Illustration of the learning problems.} For each learning problem, a classical algorithm with SQ access is exponentially faster than any quantum algorithm with quantum state inputs.  The minus sign search problem is a toy problem where the learner is given a list of $2^n$-dimensional vectors which have pairwise distance exponentially close to zero; one of the vectors contains a single negative entry as its first component, and the rest of the vectors have all positive entries.  The learner must determine which vector has the negative entry.
    In the real vector search problem, the learner is given a set of random $2^n$-dimensional vectors, for which one is real-valued and the rest are complex-valued; the learner must determine which vector is the real-valued one.}
    \label{fig:LearningProblem}
\end{figure*}

In the present work, we show that in some learning tasks, classical algorithms with SQ access can be exponentially more powerful than all quantum algorithms with quantum state inputs.
This result seems paradoxical because one would expect quantum algorithms to be at least as powerful as classical algorithms, as the former can efficiently simulate the latter.
The resolution is that SQ access gives classical algorithms substantially more information about an $x \in \mathbb{C}^{2^n}$ than can be attained from the quantum state $|x\rangle$ by an efficient quantum algorithm.
This raises the concern of whether SQ access is the appropriate classical analog of quantum state inputs.
Ideally, the classical data access model should be at most as powerful as what a quantum algorithm could obtain from its input quantum states.

{\color{black} One can also view our results from a different perspective. When the quantum state inputs used in quantum ML algorithms \cite{lloyd2014quantum, rebentrost2014quantum, harrow2009quantum, biamonte2017quantum} are constructed from classical data stored in an exponentially large classical memory, the comparison to classical algorithms with SQ access is reasonable because constructing SQ access takes classical exponential time, while encoding an exponentially long vector to the amplitude of a quantum state takes quantum exponential time.
In this scenario, quantum algorithms that have been dequantized \cite{tang2018quantum, tang2019quantum, gilyen2018quantum, chia2020sampling} still fail to provide an exponential advantage.
Furthermore, our results show that quantum algorithms with quantum state inputs may be weaker than classical algorithms due to the deficiency of quantum state inputs compared to SQ access.
}

In the remainder of this note, we first review dequantization and discuss a toy problem to illustrate several basic proof strategies and some accompanying conceptual points.
Thereafter, we dive into our main exponential separation for the \textit{real vector search problem}.
Figure~\ref{fig:LearningProblem} depicts the toy problem and the real vector problem.
We follow with remarks on other complexity-theoretic advantages afforded by SQ access, other classical accesses which are alternative to SQ access, and the relationship between SQ access and quantum state data.
We conclude with a more general discussion.

\section{A brief review on dequantization}

Before presenting our main results, we review some basic definitions in the dequantization literature~\cite{tang2018quantum, tang2019quantum, gilyen2018quantum, chia2020sampling}.
We begin with the definition of SQ access.

\begin{definition}[SQ access]
The sample and query access $\mathrm{SQ}(x)$ to an exponentially long vector $x \in \mathbb{C}^{2^n}$ is an oracle that enables the following actions:
\begin{enumerate}
    \item $\mathrm{Sample}$: The oracle outputs $i \in \{1, \ldots, 2^n\}$ with probability $|x_i|^2 / \sum_{j} |x_j|^2$.
    \item $\mathrm{Query}(i)$: The oracle outputs $x_i$ to an arbitrary precision.
    \item $\mathrm{QueryN}$: The oracle outputs $\norm{x}_2$.
\end{enumerate}
\label{def:SQaccess}
\end{definition}
\noindent Each of the above operations may be performed with unit cost.  SQ access is central to the definition of dequantization~\cite{tang2018quantum, tang2019quantum, gilyen2018quantum, chia2020sampling} given below.
\begin{definition}[Dequantization via SQ access]
Let $\cA$ be a quantum algorithm with inputs $\ket{x_1}, \ldots, \ket{x_C}$, and which outputs either a state $\ket{y}$ or a value $\lambda$.
We say that we have dequantized $\cA$ if there is a classical algorithm such that given $\mathrm{SQ}(x_1), \ldots, \mathrm{SQ}(x_C)$, the classical algorithm can evaluate queries to $\mathrm{SQ}(y)$ or output $\lambda$ with (i) a similar performance guarantee as $\cA$, (ii) a number of SQ accesses polynomial in $C$, and (iii) a runtime that is at worst polynomially slower than the runtime of $\cA$.
\end{definition}
\noindent In plainer language, we say that a quantum algorithm can be dequantized if we can find a classical algorithm using SQ access with similar performance guarantees and at most polynomial resource overhead compared to a quantum algorithm with quantum state inputs.

In previous works \cite{tang2018quantum, tang2019quantum, gilyen2018quantum, chia2020sampling}, many quantum ML algorithms, including low-rank regression, recommendation systems, and principal component analysis, have been dequantized as per the definition above.

\section{Exponential separation for a toy problem}

As a warm-up, we consider the \textit{minus sign search problem} as a toy problem which can be trivially solved via SQ access, and which is easily shown to be exponentially hard in the quantum setting.
Our purpose is to illustrate a conceptual point before moving on to our main problem of interest.
\begin{definition}[Minus sign search problem]
Fix a constant $C = \mathcal{O}(1)$ and let $d = 2^n$.
Consider $x_1, \ldots, x_C \in \mathbb{C}^{d}$, where one of the vectors $x_{k^*}$ is $(-\frac{1}{\sqrt{d}},\frac{1}{\sqrt{d}},\frac{1}{\sqrt{d}},...,\frac{1}{\sqrt{d}})$ and the rest are $(\frac{1}{\sqrt{d}},\frac{1}{\sqrt{d}},\frac{1}{\sqrt{d}},...,\frac{1}{\sqrt{d}})$.
The goal is to output the vector $x_{k^*}$, namely the vector which has a component with a minus sign.\label{def:minussearch}
\end{definition}
\noindent We immediately have the following result:
\begin{proposition}
A classical algorithm with SQ access to the set of vectors $x_1, \ldots, x_C$ can solve the minus sign search problem in constant time.
\end{proposition}
\begin{proof}
The classical algorithm simply performs Query($1$) on the first components of each of the $x_i$'s, and checks which one is negative.
The algorithm then outputs the oracle $\text{SQ}(x_{k^*})$.
\end{proof}

In short, the classical algorithm with the SQ oracle simply looks at the first component of each $x_i$ vector, and checks where the minus sign is.  However, quantum algorithms only have access to the states $| x_i\rangle$.  But the overlap between $|x_{k^*}\rangle$ and any $|x_i\rangle$ for $i \not = k^*$ is $1 - 2/d$ which is exponentially close to unity, and so we expect the minus sign search problem to be exponentially hard in the quantum setting.  Let us make this more precise and then comment on the implications. We begin with the following standard lemma.

\begin{lemma}[Two state discrimination]
A quantum algorithm can distinguish two states $\rho_1$ and $\rho_2$ with probability at least $0.9$ if and only if $\frac{1}{2} + \frac{1}{2}\|\rho_1 - \rho_2\|_1 \geq 0.9$, or equivalently $\|\rho_1 - \rho_2\|_1 \geq 0.8$.\label{lemma:twostate}
\end{lemma}

\noindent Now we can proceed to show exponential hardness for quantum algorithms with quantum state inputs:

\begin{proposition} Any quantum algorithm that can access copies of $|x_1\rangle,...,|x_C\rangle$ requires $\Omega(2^n)$-time to solve the minus sign search problem.
\label{prop:toyquantumbound1}
\end{proposition}
\begin{proof}
Consider the $C = 2$ case, since the $C > 2$ cases are at least as hard.  Let $|x_*\rangle$ be the state with the minus sign, and $|x\rangle$ be the state without the minus sign.  So either $|x_1\rangle = |x_*\rangle$ and $|x_2\rangle = |x\rangle$, or $|x_1\rangle = |x\rangle$ and $|x_2\rangle = |x_*\rangle$.  Suppose the quantum algorithm uses $N_1$ copies of $|x_1\rangle$ and $N_2$ copies of $|x_2\rangle$.  If the quantum algorithm solves the minus sign distinction problem, then it can distinguish between $|x_*\rangle \langle x_*|^{\otimes N_1} \otimes |x\rangle \langle x|^{\otimes N_2}$ and $|x\rangle \langle x|^{\otimes N_1} \otimes |x_*\rangle \langle x_*|^{\otimes N_2}$.

Let $N = \max(N_1, N_2)$.  Then it is strictly easier to distinguish between $|x_*\rangle \langle x_*|^{\otimes N} \otimes |x\rangle \langle x|^{\otimes N}$ and $|x\rangle \langle x|^{\otimes N} \otimes |x_*\rangle \langle x_*|^{\otimes N}$.  But by Lemma~\ref{lemma:twostate}, if we want to distinguish with probability at least $0.9$, then we require
\begin{align}
&\left\| |x_{*}\rangle\langle x_{*}|^{\otimes N} \!\otimes\! |x\rangle \langle x|^{\otimes N} - |x\rangle\langle x|^{\otimes N} \!\otimes\! |x_*\rangle \langle x_*|^{\otimes N} \right\|_1  \nonumber \\
& \qquad \qquad \qquad \qquad \qquad \quad = \left(1 - \frac{2}{d}\right)^{4N} \geq 0.8\,,
\end{align}
which implies $N \geq \Omega(2^n)$.  Since $N = \max(N_1, N_2)$ is lower bounded by $\Omega(2^n)$, the time complexity of the quantum algorithm is likewise lower bounded by $\Omega(2^n)$.
\end{proof}

The conceptual takeaway is that SQ access can render an $x_i$ and $x_j$ distinguishable even if their corresponding state representations $|x_i\rangle$ and $|x_j\rangle$ have an overlap very close to unity.  This is what leads to the advantage of the classical setting over the quantum setting for this toy problem.  In the next section, we provide a more surprising result that is conceptually different.

\section{Exponential separation \\ for the real vector search problem}

In the toy problem above, each pair of vectors $x_i, x_j$ has a distance exponentially close to zero, which leads to a deficiency for quantum algorithms with quantum state inputs.
One may wonder if this is the only scenario that could lead to SQ access being exponentially more powerful than quantum state inputs.
Here we consider a different learning task, the \emph{real vector search problem}, where each pair of distinct vectors $x_i, x_j$ have constant pairwise distance (which is only violated with exponentially small probability).
Nonetheless, we show that for this task, classical algorithms with SQ access are still exponentially faster than quantum algorithms with quantum state inputs.
In particular, we illustrate a constant-versus-exponential separation between the two settings.
The definition of the real vector search problem is given as follows.

\begin{definition}[Real vector search problem]
Fix a constant $C = \mathcal{O}(1)$.
Consider $x_1, \ldots, x_C \in \mathbb{C}^{2^n}$, where one of the vectors $x_{k^*}$ is sampled uniformly from the real sphere $\{x \in \mathbb{R}^{2^n} : \norm{x}_2 = 1\}$ and the rest are sampled uniformly from the complex sphere $\{x \in \mathbb{C}^{2^n} : \norm{x}_2 = 1\}$.
The goal is to output the vector $x_{k^*}$.\label{def:realsearch}
\end{definition}

When we sample vectors from a $2^n$-dimensional real/complex sphere, standard concentration inequalities show that the distance between the vectors will almost always be greater than a constant independent of $2^n$.
Despite the vectors being far away from one another,
the constant-versus-exponential complexity separation still persists.
We begin by showing that this learning problem can be solved efficiently with classical algorithms under SQ access.
In particular, we consider the output of the classical algorithm to be an SQ access to the real vector $x_{k^*}$ among the set of $C$ vectors.

\begin{proposition}
A classical algorithm with SQ access to the set of vectors $x_1, \ldots, x_C$ can solve the real vector search problem in constant time.
\end{proposition}
\begin{proof}
The classical algorithm goes though each $i = 1, \ldots, C$ and queries the first component of $x_i$, denoted by $x_{i 1}$, using one access to Query($1$). For any $i \neq k^*$, we have that $\text{Im}[x_{i 1}] \neq 0$ almost surely.
Hence, by a union bound, with probability one the only index $i$ with $\text{Im}[x_{i 1}] = 0$ is $k^*$.
Thus the classical algorithm can find $k^*$ in a time constant in $n$.
The classical algorithm then outputs $\text{SQ}(x_{k^*})$.
The total time used by the classical algorithm with SQ access is constant with respect to $n$, and only linear in $C$.
\end{proof}

Surprisingly, the real vector search problem actually requires exponential time to solve if one uses quantum algorithms with quantum state inputs.
In short, we will show that the search problem may be reduced to a quantum state discrimination problem requiring an exponential number of state copies to accomplish the task with high probability.

\begin{theorem}
Any quantum algorithm that can access copies of $\ket{x_1}, \ldots, \ket{x_C}$ requires $\Omega(2^{n/2})$-time to solve the real vector search problem.
\end{theorem}
\begin{proof}
Let $d = 2^n$, and $C = 2$.  We will prove our exponential lower bound in the $C = 2$ case, since the $C > 2$ cases are at least as hard.  Let $|x_*\rangle$ be the Haar-random real state, and let $|x\rangle$ be the Haar-random complex state. Then either $|x_1\rangle = |x_*\rangle$ and $|x_2\rangle = |x\rangle$, or $|x_1\rangle = |x\rangle$ and $|x_2\rangle = |x_*\rangle$.  Following the same proof strategy as that of Proposition~\ref{prop:toyquantumbound1}, we suppose the quantum algorithm uses $N_1$ copies of $|x_1\rangle$ and $N_2$ copies of $|x_2\rangle$.  Then if the quantum algorithm solves the distinction problem, it can distinguish between $|x_*\rangle \langle x_*|^{\otimes N_1} \otimes |x\rangle \langle x|^{\otimes N_2}$ and $|x\rangle \langle x|^{\otimes N_1} \otimes |x_*\rangle \langle x_*|^{\otimes N_2}$.  Letting $N = \max(N_1, N_2)$, it is strictly easier to distinguish between $|x_*\rangle \langle x_*|^{\otimes N} \otimes |x\rangle \langle x|^{\otimes N}$ and $|x\rangle \langle x|^{\otimes N} \otimes |x_*\rangle \langle x_*|^{\otimes N}$.  Taking
\begin{equation}
\rho_1 = \Expect_{\text{\rm Haar}(\mathbb{C}^d)} \left(|\psi\rangle\langle \psi| \right)^{\otimes N} \otimes \Expect_{\text{\rm Haar}(\mathbb{R}^d)} \left(|\phi\rangle\langle \phi| \right)^{\otimes N}
\end{equation}
and
\begin{equation}
\rho_2 = \Expect_{\text{\rm Haar}(\mathbb{R}^d)} \left(|\phi\rangle\langle \phi| \right)^{\otimes N} \otimes \Expect_{\text{\rm Haar}(\mathbb{C}^d)} \left(|\psi\rangle\langle \psi| \right)^{\otimes N}\,,
\end{equation}
Lemma~\ref{lemma:twostate} tells us that if we want to distinguish with probability at least 0.9 then we need
\begin{equation} \label{eq:rho1-rho2-0.8}
    \|\rho_1 - \rho_2\|_1 \geq 0.8\,.
\end{equation}
Using the triangle inequality, we have
\begin{align}
\label{E:1normtobound}
    &\norm{\rho_1 - \rho_2}_1 \nonumber \\
    &\leq \!2\,\Big\|\!\Expect_{\text{\rm Haar}(\mathbb{C}^d)}\! \!\left(|\psi\rangle\langle \psi| \right)^{\otimes N} \!\!-\! \Expect_{\text{\rm Haar}(\mathbb{R}^d)}\!\! \left(|\phi\rangle\langle \phi| \right)^{\otimes N}\!\Big\|_1
\end{align}
and so it remains to bound the 1-norm quantity above.

Recall that
\begin{align}
&\Expect_{\text{\rm Haar}(\mathbb{R}^d)}\! \left(|\phi\rangle\langle \phi| \right)^{\otimes N} \!=\! \frac{1}{\left(\!\frac{2^N\, \Gamma(N+d/2)}{\Gamma(d/2)}\!\right)} \!\sum_{M \in \mathcal{M}_{2N}} \!\!\!\sigma_M
\end{align}
where $\cM_{2N}$ is the group of pair permutations and $\sigma_M$ is their representation on $(\mathbb{C}^d)^{\otimes N}$, and
\begin{align}
&\Expect_{\text{\rm Haar}(\mathbb{C}^d)} \left(|\psi\rangle\langle \psi| \right)^{\otimes N} = \frac{1}{\binom{d+N-1}{N}} \,P_{\vee^N}\,
\end{align}
where $P_{\vee^N}$ is the projection onto the completely symmetric subspace over $(\mathbb{C}^d)^{\otimes n}$ (see~\cite{harrow2013church} for a review).
From the definition of pair permutations, we have
\begin{align}
&\sum_{M \in \mathcal{M}_{2N}} \sigma_M = N!\,P_{\vee^N} \nonumber \\
&\qquad + \sum_{1 \leq k \leq N/2} c_k \, P_{\vee^N} \cdot \left(\Phi^{\otimes k} \otimes \mathds{1}^{\otimes(N-2k)}\right) \cdot P_{\vee^N}
\end{align}
where $\Phi = \frac{1}{2}(|00\rangle + |11\rangle)(\langle 00| + \langle 11|)$ is the density matrix of a Bell state, $\mathds{1}$ is the identity on one qubit, and the $c_k$ are combinatorial coefficients which are each greater than or equal to zero.
Therefore, we can write
\begin{align}
\Expect_{\text{\rm Haar}(\mathbb{R}^d)}\! \left(|\phi\rangle\langle \phi| \right)^{\otimes N} \!\!=\! \frac{N!}{\left(\!\frac{2^N\,\! \Gamma(N+d/2)}{\Gamma(d/2)}\!\right)} \, P_{\vee^N} \!+\! O_{\text{rest}},
\end{align}
where $O_{\text{rest}}$ is positive semi-definite.  Then we have
\begin{align}
&\left\|\Expect_{\text{\rm Haar}(\mathbb{C}^d)} \left(|\psi\rangle\langle \psi| \right)^{\otimes N} - \Expect_{\text{\rm Haar}(\mathbb{R}^d)} \left(|\phi\rangle\langle \phi| \right)^{\otimes N} \right\|_1 \nonumber \\ &\leq \left|\frac{N!}{\left(\frac{2^N\, \Gamma(N+d/2)}{\Gamma(d/2)}\right)} - \frac{1}{\binom{d+N-1}{N}} \right| \, \|P_{\vee^N}\|_1 + \|O_{\text{rest}}\|_1 \nonumber \\
&= \left|1 \!-\! \frac{(d+N-1)! \,(d/2-1)!}{2^N (d/2 + N-1)! \,(d-1)!}\right| + \text{tr}(O_{\text{rest}}) \nonumber\\
&\leq \left|1 \!-\! \frac{1}{\left(1 \!+\! 2N/d \right)^N} \right| \!+\! \left(1 \!-\! \frac{(d\!+\!N\!-\!1)! \,(d/2\!-\!1)!}{2^N (d/2\!+\!N-1)! \,(d\!-\!1)!} \right) \nonumber \\
&\leq 2 \left(1 - \frac{1}{\left(1 + 2N/d \right)^N}\right) \leq \frac{4 N^2}{d}\,.
\end{align}
Combining this with Eq.~\eqref{E:1normtobound}, we see that
\begin{equation}
    \norm{\rho_1 - \rho_2}_1 \leq \frac{4 N^2}{d}.
\end{equation}
From Eq.~\eqref{eq:rho1-rho2-0.8}, we have $N \geq \Omega(2^{n/2})$.

Since the quantum algorithm must use at least $N = \max(N_1,N_2) = \Omega(2^{n/2})$ state copies to solve the real vector search problem, and we assume the preparation of one of these states has unit cost (analogous to one use of the SQ oracle), the quantum algorithm must run in $\Omega(2^{n/2})$-time.
\end{proof}

% What is interesting about the real vector search problem is that if the set of states $|x_1\rangle,...,|x_C\rangle$ are fixed, then it is easy to identify ${k^*}$ because their pairwise distance are constantly large.
% This is different than the setting of the minus sign search problem: even if the set of states $|x_1\rangle,...,|x_C\rangle$ are fixed, identifying $k^*$ is still very hard.

% The real vector search problem illustrates a key point that some property encoded in the states $|x_1\rangle,...,|x_C\rangle$, here whether a state is complex-valued or not, can be exponentially hard to extract from copies of the quantum states but very easy to extract when we have SQ access to the states.

\section{SQ access and \#P}

As a further illustration of the care that must be taken when invoking an SQ access model with respect to quantum states, we provide another observation: obtaining SQ access to a quantum state generated by a polynomial-time quantum circuit is at least as hard as solving \#P-complete problems.
Given a polynomial-time quantum circuit $U$ on $n$ qubits, we can consider the state
\begin{equation}
     \ket{\psi_U} = (\Id \otimes U^\dagger) \mathrm{CNOT}_{2 \rightarrow 1} (\Id \otimes U) (\ket{0} \otimes \ket{0^n})\,,
\end{equation}
where $\mathrm{CNOT}_{2 \rightarrow 1}$ is the CNOT gate controlled by the 2nd qubit and acting on the 1st qubit.
We readily see that $\braket{0^{n+1} | \psi_U}$ is equal to the probability of obtaining the outcome zero when we measure the first qubit of $U|0^n\rangle$.
If we have SQ access to $\ket{\psi_U}$, then we can consider $\mathrm{Query}(0 \ldots 0)$, which outputs an arbitrarily accurate estimate for $\braket{0^{n+1} | \psi_U}$.
Therefore, having SQ access to $\ket{\psi_U}$ enables us to achieve a strong quantum simulation of $U$.
Strong quantum simulation of a polynomial-time quantum circuit \cite{nest2008classical} is as hard as solving \#P-Complete problems, which implies the claimed result.

\section{Other types of classical access}

Our results establishing exponential separations between classical algorithms with SQ access and quantum algorithms with quantum state inputs suggest that SQ access may not be the appropriate classical analog of quantum state inputs.
In particular, SQ access to the vectors $x_1,...,x_C$ allow the classical algorithm to manipulate information which is exponentially hard to access from the states $|x_1\rangle,...,|x_C\rangle$.
Equivalently, any quantum algorithm which takes copies of the states $|x_1\rangle,...,|x_C\rangle$ and produces an oracle implementing SQ access requires at least exponential time complexity.

In order to modify the definition of SQ access to a putative $\widetilde{\text{SQ}}$ access which precludes the aforementioned pathologies, let us suggest a criterion:
\begin{criterion}
Give copies of a state $|x\rangle = \sum_i x_i |i\rangle$, there is a polynomial-time quantum algorithm to simulate the oracle $\widetilde{\text{\rm SQ}}(x)$.
\end{criterion}
\noindent For instance, we could let $\widetilde{\text{SQ}}$ access just have the sample operation.  Or alternatively, we could allow $\widetilde{\text{SQ}}$ to access outcomes of a polynomial-time POVM measurement applied to $|x\rangle$ (see, for instance,~\cite{huang2021information, aharonov2021quantum, chen2021exponential, huang2021demonstrating}).

\section{Comments on quantum state data}

Thus far, we have focused on how SQ access to a set of classical vectors $x_1,...,x_C$ provides much more information than can be learned from the quantum states $|x_1\rangle,...,|x_C\rangle$.  We can flip this logic around: the quantum states $|x_1\rangle,...,|x_C\rangle$ provide much less information than the set of classical vectors $x_1,...,x_C$.  To illustrate this point, consider the following modification of the minus sign search problem:
\begin{definition}[Unnormalized minus sign search problem]
Fix a constant $C = \mathcal{O}(1)$ and let $d = 2^n$.
Consider $x_1, \ldots, x_C \in \mathbb{C}^{d}$, where one of the vectors $x_{k^*}$ is $(-1,1,1,...,1)$ and the rest are $(1,1,1,...,1)$.
The goal is to output the vector $x_{k^*}$, namely the vector which has a component with a minus sign.\label{def:unnormalizedminussearch}
\end{definition}
Clearly this is trivial to solve with SQ access; the proof is the same as in the original minus sign search problem.  However, the point is that if we were to store the classical data as copies of quantum states $|x_1\rangle,...,|x_C\rangle$ whose components are multiplied by $1/\sqrt{d}$ factors on account of normalization, then the quantum algorithm will not be able to solve the minus sign search problem.
This is because the mapping of the classical vector $(-1,1,1,...,1)$ to the quantum state $(-\frac{1}{\sqrt{d}},\frac{1}{\sqrt{d}}, \frac{1}{\sqrt{d}},...,\frac{1}{\sqrt{d}})$ is not the appropriate classical-to-quantum encoding.  The more natural encoding is along the lines of
\begin{equation}
(-1,1,1,...,1) \mapsto |-\rangle \otimes |+\rangle \otimes |+\rangle \otimes \cdots \otimes |+\rangle\,.
\end{equation}
If this mapping were to be performed (and similarly for the vectors of all ones), then the quantum algorithm would be just as efficient as the classical algorithm with SQ access.

What the above highlights is that, while the QML literature often considers embedding classical vectors into the amplitudes of quantum states,
this may sometimes cause quantum algorithms to be exponentially slower than classical algorithms.

\section{Discussion}

We have shown that classical algorithms with SQ access can sometimes be exponentially more powerful than quantum algorithms with quantum state inputs.
{\color{black}This seemingly contradictory result shows that SQ access can be exponentially more powerful than quantum state inputs.}
At the same time, our results imply that constructing SQ access from copies of a quantum state $\rho$ requires exponentially many copies of $\rho$, an exponential number of measurements, and hence an exponential amount of time.

In order to avoid results stating that classical algorithms are more powerful than quantum algorithms, we could consider classical algorithms with access to an oracle that is at most as powerful as inputs given to the quantum algorithms.
For example, when comparing to quantum algorithms with quantum state inputs, we could consider classical algorithms with access to classical data obtained by measuring the quantum states.
Classical algorithms with access to measurement data are still very powerful, being capable of predicting outcomes of quantum experiments \cite{huang2021information, aharonov2021quantum, chen2021exponential, huang2021demonstrating}, classifying quantum phases of matter \cite{huang2021provably}, predicting ground state properties \cite{huang2021provably}, etc.
However, classical algorithms with measurement data access will never be more powerful than quantum algorithms with quantum state inputs because quantum algorithms can always perform the measurements within the algorithm.
This would then avoid perplexing statements such as that classical algorithms can be exponentially faster than quantum algorithms.

An important future direction is to study how quantum advantage becomes evident when we replace classical algorithms leveraging the powerful SQ access with measurement data access.
It is likely that various learning tasks considered to have no exponential quantum advantage are thought of as such due to the power of SQ access.
For example, we established an exponential quantum advantage for quantum principal component analysis (quantum PCA) in~\cite{huang2021demonstrating} when we compare quantum algorithms with quantum state inputs to classical algorithms with access to measurement data.
This result contrasts with the lack of exponential advantage in quantum PCA~\cite{tang2018quantum} when we compare to classical algorithms with SQ access.
By comparing quantum algorithms with quantum state inputs to classical algorithms with access to measurement data,
we are hopeful that new and significant quantum advantages could be established, and that the grounds for claiming such quantum advantages will be made clearer.

{\color{black} Finally, consider the setting where the quantum state $\ket{x} = \sum_{i=1}^d x_i \ket{i}$ is obtained from an exponentially long classical vector $x \in \mathbb{C}^d$ stored in classical RAM (random access memory).
In this setting, constructing an SQ access is less demanding than constructing quantum state inputs (although both require an exponential amount of time). Hence, it is reasonable to compare classical algorithms with SQ access to quantum algorithms with quantum state inputs when the quantum states are obtained from classical data.
We emphasize that in this setting, quantum algorithms that have been dequantized \cite{tang2018quantum, tang2019quantum, gilyen2018quantum, chia2020sampling} do not yield an exponential speedup.
Furthermore, our results show that classical algorithms could be significantly more powerful than quantum algorithms with quantum states that encode classical data in the amplitudes.
}

\subsection*{Acknowledgments:}
\vspace{-0.5em}

We thank Scott Aaronson and Ewin Tang for valuable discussions, and thank Sam McArdle for catching typos.
JC is supported by a Junior Fellowship from the Harvard Society of Fellows, the Black Hole Initiative, as well as in part by the Department of Energy under grant {DE}-{SC0007870}. HH is supported by a Google PhD Fellowship.

\bibliography{references}
\bibliographystyle{abbrv}

\end{document}